\documentclass[10pt]{article}
     \usepackage{fullpage}

\usepackage{amssymb}
\usepackage{graphicx}
\usepackage{amsmath}
\usepackage{algorithm}
\usepackage{algorithmic}

\usepackage{comment}

\newcommand{\ALG}{\textsc{ACE}}
\newcommand{\OFF}{{\bf OFF}}

    \newtheorem{theorem}{Theorem}[section]
    \newtheorem{lemma}[theorem]{Lemma}

   \def\qed{\hspace*{\fill}$\Box$}

\usepackage{url}

\begin{document}

\date{}



\title{Optimal In-Network Function Placement\thanks{Supported by the FP7 EU project UNIFY, the German-Israeli GIF project I-1245-407.6/2014, and ...}
\\ {\small An $O(\log{n})$-Approximation for Generalized Capacitated Set Cover}%
}

\title{Placing and Routing via Function Chains\\ {\small An $O(\log{n})$-Approximation of Generalized Capacitated Set Cover}%
\thanks{Supported by the FP7 EU project UNIFY, the German-Israeli GIF project I-1245-407.6/2014, and ...}
}

\title{Minimum Network Function Placement\\ {\small $O(\log{n})$-Approximation for Generalized Capacitated Set Cover}%
\thanks{Supported by the FP7 EU project UNIFY and the German-Israeli GIF project I-1245-407.6/2014.}
}

\title{The Online and Approximate\\Service Chain Embedding Problem%
\thanks{Supported by the FP7 EU project UNIFY and the German-Israeli GIF project I-1245-407.6/2014.}
}

\title{Online Admission Control and\\Embedding of Service Chains%
\thanks{Supported by the FP7 EU project UNIFY and the German-Israeli GIF project I-1245-407.6/2014.}
}

\author{Tam\'as Lukovszki$^1$ and 
 Stefan Schmid$^{2}$\\ 
   {\small
$^1$Faculty of Informatics,
  E\"otv\"os Lor\'and University, Budapest, Hungary}\\
{\small $^2$TU Berlin \& Telekom Innovation Laboratories, Berlin, Germany}
}

\newcommand{\capa}{\kappa}

\maketitle

\sloppy

\begin{abstract}
The virtualization and softwarization
of modern computer networks enables
the definition and fast deployment of novel
network services called \emph{service chains}:
sequences of
virtualized network functions (e.g., firewalls, caches, traffic optimizers)
through which traffic is routed between source and destination.
This paper attends to the problem of
admitting and embedding a maximum number of service chains, i.e.,
a maximum number of source-destination pairs which are routed
via a sequence of ${\ell}$ to-be-allocated, capacitated network functions.
We consider an Online variant of this
maximum Service Chain Embedding Problem, short \emph{OSCEP},
where requests arrive over time, in a worst-case manner.
Our main contribution is a deterministic $O(\log \ell)$-competitive
online algorithm, under the assumption that capacities are at least logarithmic in $\ell$.
We show that this is asymptotically optimal within the class of deterministic
and randomized online algorithms.
We also explore lower bounds for offline approximation algorithms, and prove that
the offline problem is APX-hard for unit 
capacities and small~$\ell\geq 3$,
and even Poly-APX-hard in general, when there is no bound on~$\ell$.
These approximation lower bounds may be of independent interest,
as they also extend to other problems such as Virtual Circuit Routing.
Finally, we present an exact algorithm based on 0-1 programming, implying that the general
offline SCEP is in NP and by the above hardness results it is NP-complete for constant $\ell$.
\end{abstract}

\section{Introduction}\label{sec:intro}

Today's computer networks provide a rich set of in-network functions,
including access control, firewall, intrusion detection, network address translation,
traffic shaping and optimization,
caching, among many more.
While such functionality is traditionally implemented in hardware middleboxes,
computer networks become more and more
virtualized~\cite{opennf,manifesto}:
\emph{Network Function Virtualization (NFV)} enables a flexible
instantiation of network functions on
network nodes, e.g., running in a virtual machine
on a commodity x86 server.

Modern computer networks also offer new flexibilities in terms of how
traffic can be routed through such network functions.
In particular, using \emph{Software-Defined Networking (SDN)}~\cite{openflow}
technology,
traffic can be steered along arbitrary routes, i.e.,
along routes which depend on the application~\cite{sdx}, and which
are not necessarily shortest paths or destination-based,
or
not even loop-free~\cite{flowtags}.

These trends enable the realization of
interesting new in-network communication services
called \emph{service chains}~\cite{ETSI,ewsdn14,merlin}:
sequences of network functions which are allocated and stitched
together in a flexible manner.
For example, a service chain $c_i$ could define that
the traffic originating at source $s_i$ is first steered through
an intrusion detection system
for security (1$^{\mathit{st}}$ network function),
next through a traffic optimizer (2$^{\mathit{nd}}$ network function),
and only then is routed towards the destination $t_i$.
%
Such advanced network services open an interesting new
market for Internet Service Providers, which can become
``miniature cloud providers''~\cite{eurosys15} specialized for
in-network processing.

\subsection{Paper Scope}

In this paper, we study the problem of how to optimally admit
and embed
service chain requests.
Given a redundant distribution of network functions and a sequence $\sigma=(\sigma_1,\sigma_2,\ldots,\sigma_k)$,
where each $\sigma_i=(s_i,t_i)$ for $i\in[1,k]$
defines a source-destination pair $(s_i,t_i)$ which needs
to be routed via a sequence of network function instances, we ask:
Which requests $\sigma_i$
to admit
and where to allocate their service chains $c_i$?
The service chain embedding should respect capacity constraints as
well as constraints on the length (or stretch) of the route from $s_i$
to $t_i$ via its service chain $c_i$.

Our objective is to maximize the number of admitted
requests. We are particularly interested in the
\emph{Online Service Chain Embedding Problem (OSCEP)},
where $\sigma$ is only revealed over time.
We assume that a request cannot be delayed
 and once admitted, cannot be preempted again.
 Sometimes, we are also interested in the general (offline) problem,
 henceforth denoted by SCEP.


\subsection{Our Contribution}

We formulate the online and offline problems OSCEP
and SCEP, and make the following contributions:
\begin{enumerate}
\item We present a deterministic online algorithm $\ALG$\footnote{\textbf{A}dmission control and \textbf{C}hain \textbf{E}mbedding.} which, given
that node capacities are at least logarithmic, achieves
a competitive ratio $O(\log{\ell})$ for OSCEP. This result is practically interesting,
as the number of to be traversed network functions $\ell$ is likely to be small in practice.
To the best of our knowledge,
so far, only heuristic and offline approaches
to solve the service chain embedding problem have been considered~\cite{DBLP:journals/corr/BariCAB15,karl-chains,merlin}.

\item We establish a connection to virtual circuit routing and
prove that $\ALG$ is asymptotically optimal in the class of both
deterministic and randomized online algorithms.
Moreover, we initiate the study of lower bounds
for the offline version of our problem, and show that no good
approximation algorithms exist, unless
$P=NP$: for unit capacities and already small $\ell$,
the offline problem SCEP is APX-hard. For arbitrary $\ell$,
the problem can even become Poly-APX-hard.
These results
also apply to the offline version of classic online call control problems,
which to the best of our knowledge have not been studied before.

\item
We present a 0-1 program for SCEP, which
also shows
that
SCEP is in NP for constant $\ell$ and, taking into account our hardness result,
that SCEP is NP-complete for constant $\ell$.
More precisely, if the number of all possible chains
that can be constructed over the network function instances is polynomial in
the network size $n$ then the number of variables in the 0-1 program is
also polynomial, and thus the problem is in NP.
If $m_i$ is the number of instances of network function $f_i$ in the network,
$i= 1,...,\ell$,
and $m=\max_i\{m_i\}$, then the size of the 0-1 program
is polynomial for $m^\ell = \textrm{poly}(n)$.
This always holds for constant $\ell$.
When $m$ is constant, then it holds for $\ell = O(\log n)$.
\end{enumerate}

\subsection{Outline}

This paper is organized as follows.
Section~\ref{sec:model} introduces our model
and puts the model into perspective with respect to classic online
optimization problems.
We present the 0-1 program in Section~\ref{sec:01}:
the section also serves as a formal model for our problem.
Section~\ref{sec:algoanalysis} presents and analyzes
the $O(\log{\ell})$-approximation algorithm, and
Section~\ref{sec:lowerbound}
presents our lower bound.
We summarize our results and conclude our work in
Section~\ref{sec:summary}.

\section{Model}\label{sec:model}

We are given an undirected network $G=( V,E)$ with
$n=|V|$ nodes and $m=|E|$ edges.
On this graph, we need to route a sequence of requests $\sigma=(\sigma_1,\sigma_2, \ldots, \sigma_k)$:
$\sigma_i$ for any $i$ represents a node pair $\sigma_i=(s_i,t_i) \in V\times V$.
Each pair $\sigma_i$
needs to be routed (from $s_i$ to $t_i$) via 
a sequence of $\ell$ network functions $(F_1,\ldots,F_{\ell})$.
For each network function type
$F_i$, there exist multiple instantiations 
$f_i^{(1)}, f_i^{(2)}, \ldots$ in the network. (We will omit the superscript if it is irrelevant or
clear in the context.)
Each of these instances can be applied to $\sigma_i$ along the route
from $s_i$ to $t_i$. However, in order to minimize the detour
via these functions and in order to keep the route from
$s_i$ to $t_i$ short, a ``nearby instance'' $f_i^{(j)}$ should be chosen,
for each $i$. 
A service chain instance for $(s_i,t_i)$ is denoted by
 $c_i=(f_1^{(x_1)},f_2^{(x_2)},\ldots,f_{\ell}^{(x_{\ell})})$,
for some function instances $f_j^{(x_y)}$, $j\in[1,{\ell}]$.

%
%

For ease of presentation, we will initially assume that
requests $\sigma_i$ are of infinite duration. We will later show how
to generalize our results to scenarios where requests can have arbitrary
and unknown durations.

\begin{figure}[ht]
\centering
\includegraphics[width=0.4\columnwidth]{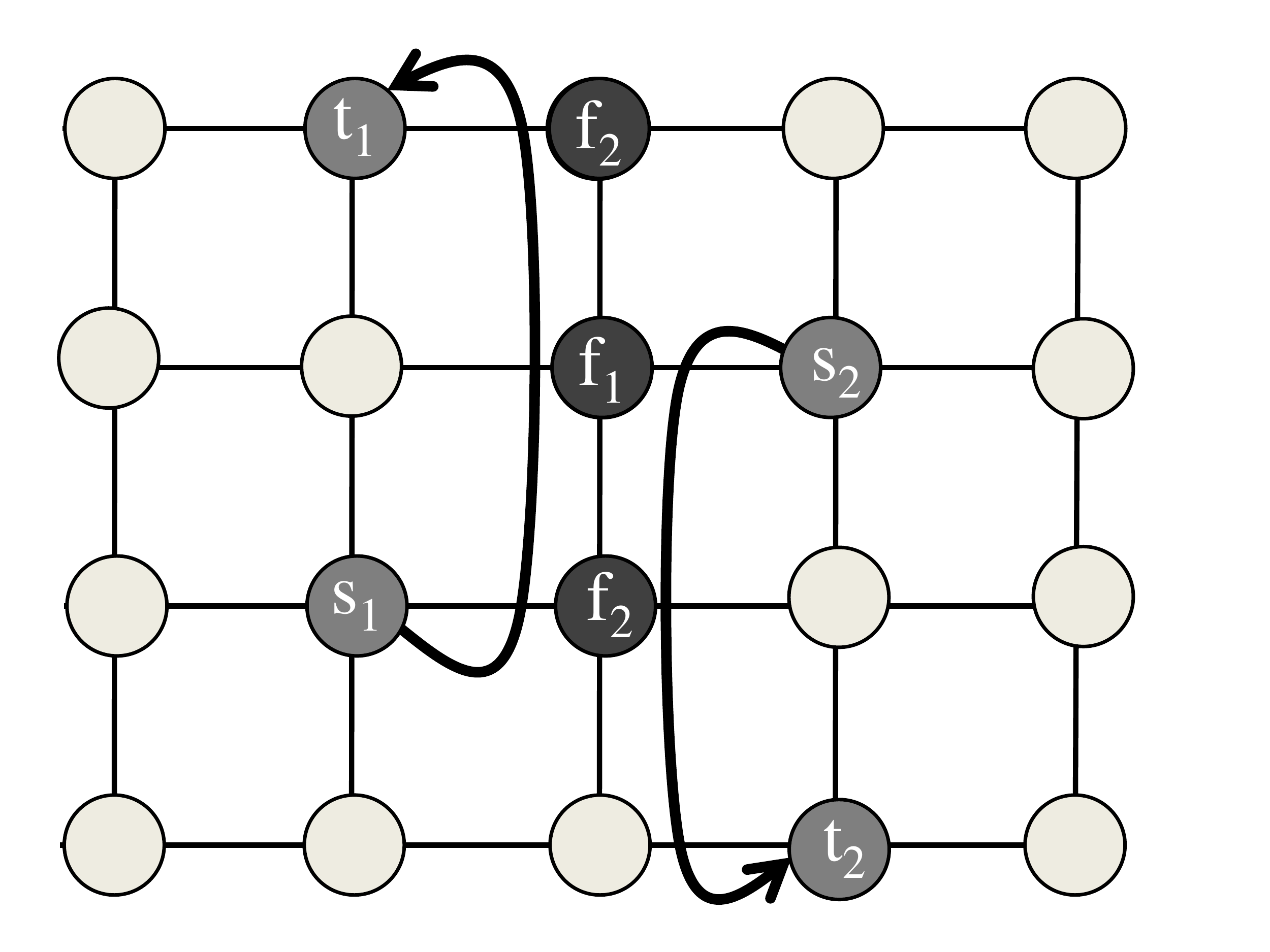}
\caption{Illustration of the model: The communication from $s_1$ to $t_1$
and from $s_2$ to $t_2$ needs to be routed via a service chain $(F_1,F_2)$.
In this example, function $F_1$ is instantiated once, and function $F_2$ is
instantiated twice. Resources for $(s_1,t_1)$ are allocated only at the
second instance of $F_2$ (the upper one).}
\label{fig:model}
\end{figure}

Concretely, in order to satisfy a request $\sigma_i=(s_i,t_i)$, a route of the following form must
be computed:
\begin{enumerate}
\item The route must start at $s_i$, traverse a sequence of network functions
$(f_1^{(x_1)},f_2^{(x_2)},\ldots,f_{\ell}^{(x_{\ell})})$, and end at $t_i$. Here, $f_j^{(x_y)}$, $j\in[1,{\ell}]$
is an instance of the network function
of type $F_j$.
\item The route must not violate capacity constraints on any node $v\in V$.
Nodes $v\in V$ are capacitated
and resources need to be allocated for each network function which is used
for any $(s_i,t_i)$ pair. Multiple network functions may be available on the same
physical machine, and only consume resources once they are used in certain
service chains. 
The capacity $\kappa(v)$ of
each node $v\in V$ hence defines the maximum number of requests $\sigma_i$
for which $v$ can apply its network functions. However,
node $v$ can always simply serve as a regular forwarding node for other requests,
without applying the
function. 
\item The route should be of (hop) length at most $r$ (or have a bounded stretch, see Section~\ref{ssec:stretch}).
\end{enumerate}
Otherwise, a request $\sigma_i$ must be rejected.
For ease of notation, in the following, we will sometimes assume that for a rejected
request $\sigma_i$,
$c_i=\emptyset$.
Also note that the resulting route may not form a simple path, but more generally
describes a \emph{walk}: it may contain
forwarding loops (e.g., visit a network function and come back).

Our objective is to maximize the number of satisfied requests $\sigma_i$,
resp.~to embed a maximum number of service chains.
We are mainly interested in the online variant of the problem,
where $\sigma$ is revealed over time.
More precisely, and as usual in the realm of online algorithms and
competitive analysis, we seek to devise an online algorithm which minimizes
the so-called \emph{competitive ratio}:
Let $\textbf{ON}(\sigma)$ denote the cost
of a given online algorithm for $\sigma$ and let
$\textbf{OFF}(\sigma)$ denote the cost of an optimal offline algorithm.
The competitive ratio $\rho$ is
defined as the worst ratio (over all possible $\sigma$) of the cost of
$\textbf{ON}$ compared to $\textbf{OFF}$. Formally, $\rho = \max_{\sigma}
\textbf{ON}(\sigma)/\textbf{OFF}(\sigma)$.

Note that solving
this optimization problem consists of two subtasks:
\begin{enumerate}
\item \emph{Admission control:} Which requests $\sigma_i$ to admit,
and which to reject?
\item \emph{Assignment and routing:} We need to assign $\sigma_i=(s_i,t_i)$ pairs to a sequence
of network functions and route the flow through them accordingly.
\end{enumerate}

See Figure~\ref{fig:lb} for an illustration of our model.

\subsection{Putting the Model into Perspective}

From an algorithmic perspective, the models closest to ours occur in the context of online call admission
respectively virtual circuit routing.
There, the fundamental problem is to decide, in an online manner, which
``calls'' resp.~``virtual circuits'' or entire networks, to admit and how to route them, in
a link-capacitated graph.~\cite{AwerbuchAP93,AwerbuchAPW94,DBLP:conf/sirocco/EvenM13,tcs12vnet,Plotkin95}
%
%

Instead of routes, in our model, service functions have to
be allocated and connected to form service chains.
In particular, in our model, nodes have a limited capacity
and can only serve as
network functions for a bounded number of source-destination pairs.
The actual routes taken in the network play a secondary
role, and may even contain loops. In particular, our model
supports the specification of explicit constraints on the length
of a route, but also on the stretch: the factor by which the length
of a route
from a source to a destination can be increased due to the need
to visit certain network functions.

Nevertheless, as this paper shows, several techniques from
classic literature on online call control can be applied to our
model. At the same time, to the best of our knowledge, some of our results
also provide new insights into the classic variants
of call admission control. For example, our lower bounds on the
approximation ratio also translate to classic problems,
which so far have mainly been studied from an online perspective.

\section{Optimal 0-1 Program and NP-Completeness}\label{sec:01}

SCEP can be formulated
as a 0-1 integer linear program.
This together with our hardness results also proves NP-completeness for constant $\ell$:
0-1 integer linear programming is one of Karp's
NP-complete problems~\cite{Karp72}.

Let $\sigma=\{\sigma_i = (s_i,t_i) : s_i,t_i \in V \}$ be the set of
requests, and
let $\mathcal{C}$ be the set of possible chains over $k$ nodes,
respecting route length constraints.
We refer by $c\in \mathcal{C}$ to a potential chain.
For all potential chains $c\in \mathcal{C}$, let $S_c$ be the
set of connection requests in $\sigma$ that can be routed through $c$ on a path
of length at most $r$, i.e., for $c=(v_1,...,v_\ell)$, let
$S_c = \{\sigma_i=(s_i,t_i)\in \sigma$ : $d(s_i,v_1)+\sum_{i=2}^k d(v_{i-1},v_{i})+d(v_k,t_i)\leq r\}$,
where $d(u,v)$ denotes the length of the shortest path between
nodes $u,v\in V$ in the network $G$.
The shortest paths between nodes can be computed in a preprocessing step.

For all connection requests $\sigma_i\in \sigma$, we introduce the  binary variable $x_{i}\in \{0,1\}$.
The variable $x_i=1$ indicates that the request $i$ is admitted in the solution.
For all potential network function chains $c\in\mathcal{C}$, we introduce the
binary variable $x_{c}\in \{0,1\}$.
The variable $x_{c}=1$ indicates that $C$ is selected in the solution.
For all $c\in \mathcal{C}$ and $\sigma_i\in \sigma$, we introduce the binary
variable $x_{c,i}\in \{0,1\}$.
The variable $x_{c,i}$ indicates that the request $\sigma_i=(s_i,t_i)\in \sigma$
is routed through the nodes of $c$, such that the length of the walk from
$s_i$ to $t_i$ through $c$ has length at most~$r$.

\begin{eqnarray} 
\textrm{maximize}\qquad  \sum_{\sigma_i\in \sigma} x_{i} & &  \label{P01_1}\\
\textrm{s.t.}\qquad x_i - \sum_{c\in\mathcal{C}} x_{c,i} = 0 & &\qquad \forall\ \sigma_i\in \sigma \label{P01_2} \\
\sum_{c\in\mathcal{C}: \sigma_i\not\in S_{c}} x_{c,i} = 0  & &\qquad \forall\ \sigma_i\in \sigma \label{P01_3}\\
x_c \leq x_v & & \qquad \forall\ v\in V, \forall\ c\in\mathcal{C}: v\in c \label{P01_4}\\
\sum_{c\in\mathcal{C}: v\in c} x_c \geq x_v & & \qquad \forall\ v\in V \label{P01_5}\\
\sum_{\sigma_i\in \sigma}\sum_{c\in\mathcal{C}: v\in c} x_{c,i} \leq \kappa(v)\cdot x_v  & &\qquad \forall\ v\in V \label{P01_6}\\
x_i, x_{v}, x_c, x_{c,i}\in\{0,1\}  & &\qquad \forall\ v\in V, \forall\ c\in\mathcal{C}, \forall\ \sigma_i\in \sigma
\end{eqnarray}

The objective function (\ref{P01_1}) asks for
admitting a request set of maximum cardinality.
The Constraints (\ref{P01_2}) enforce that each admitted request $\sigma_i\in \sigma$
is assigned to exactly one chain $c\in\mathcal{C}$, and rejected requests are not assigned
to
any chain, i.e., for each $\sigma_i$ with  $x_i=1$, there is exactly one chain $c$
with $x_{c,i}=1$, and for each $i$ with $x_i=0$, we have $x_{c,i}=0$ for all~$c$.
Constraints (\ref{P01_3}) state that each $\sigma_i\in \sigma$ can
only be assigned to a chain $c\in \mathcal{C}$ with $\sigma_i\in S_{c}$.
By definition of $S_c$, the nodes $s_i$ and $t_i$ can be routed through $c$
by a path of length at most~$r$.
Constraints (\ref{P01_4}) ensure that if a node $v\in V$ is contained in a selected chain $c$
(i.e., $x_c$ = 1), then $x_v = 1$.
Constraints (\ref{P01_5}) enforce that if a node $v\in V$ is not contained in any selected chain,
i.e., $x_c = 0$ for all chains $c$ with $v\in c$, then $x_v = 0$.
Therefore, Constraints (\ref{P01_4}) and (\ref{P01_5}) together imply that
$x_v=1$ iff $v$ is contained in a selected chain~$c$.
Constraints (\ref{P01_6}) describe that the number of requests routed through
a node $v$ of a selected chain is limited by the capacity $\kappa(v)$ of $v$.
Furthermore, (\ref{P01_6}) ensures that if $v$ is not contained in any selected chain
(i.e., $x_v=0$) then no request $q$ is assigned to any chain $c$ with $v\in c$.

The solution of this 0-1 program defines a maximum cardinality
set of admitted requests $\sigma_{admit}=\{\sigma_i : x_i = 1\}$, and an assignment of
each request $\sigma_i\in \sigma_{admit}$ to a chain $c\in \mathcal{C}$.
Each request $\sigma_i\in \sigma_{admit}$ is assigned to a chain $C\in \mathcal{C}$ iff $x_{c,i} = 1$.
This assignment guarantees that
($i$) the request $\sigma_i=(s_i,t_i)$ can be routed through $c$ on a path of length at most $r$,
($ii$) the number of pairs routed through any node $v\in V$ of a selected chain
is limited by the capacity $\kappa(v)$ of $v$, and
($iii$) none of the requests $\sigma_i\in \sigma_{admit}$ are assigned to a non selected chain.
Furthermore, it is guaranteed that rejected requests $\sigma_i\in \sigma\setminus \sigma_{admit}$
are not assigned to any chain.

\section{Competitive Online Algorithm}\label{sec:algoanalysis}

We present an online algorithm $\ALG$ for
OSCEP. $\ALG$ admits and embeds at least a $\Omega(\log {\ell})$-fraction
of the number of requests embedded by an optimal offline algorithm $\OFF$.

Let us first introduce some notation. Let $A_j$ be the set of indices of the
requests admitted by $\ALG$ just
\emph{before} considering the $j$th
request $\sigma_j$. The index set of all admitted requests after processing all
$k$ requests in $\sigma$, will be denoted by $A_{k+1}$ resp.~$A$.

The relative load $\lambda_v(j)$ at node $v$
before processing the $j$th request, is defined by the number of
service chains $c_i$ in which $v$ participates, divided
by $v$'s capacity:
$$
\lambda_v(j) = \frac{|\{c_i\ :\ i\in A_j, v\in c_i\}|}{\kappa(v)}.
$$
We seek to ensure the invariant that capacity constraints are enforced at each node,
i.e., $\forall\ v\in V, j\leq k+1 : \lambda_v(j) \leq 1$.

We define $\mu=2 {\ell}+2$, and in the following, will assume that
\begin{eqnarray}
\min_v\{\kappa(v)\} \geq \log\mu \label{label:mincap}
\end{eqnarray}

\subsection{Algorithm}

The key idea of $\ALG$ is to assign to each node, a cost which is exponential in
the relative node load. More precisely,
with each node we associate a cost $w_v(j)$ just before processing the $j$th
request $\sigma_j$: 
$$
w_v(j) = \kappa(v)(\mu^{\lambda_v(j)}-1).
$$

Our online algorithm $\ALG$ simply proceeds as follows:
\begin{itemize}
\item When request $\sigma_j$ arrives, $\ALG$ checks if there exists a chain $c_j$,
$\sigma_j\in S_{c_j}$, satisfying the following condition:
\begin{eqnarray}
\sum_{v\in c_j}\frac{w_v(j)}{\kappa(v)} \leq {\ell} \label{eqn:admit}
\end{eqnarray}
\item If such a chain $c_j$ exists, then \emph{admit}~$\sigma_j$ and
assign it to $c_j$. Otherwise, reject~$\sigma_j$.
\end{itemize}

In order to ensure that chains selected for Condition~\ref{eqn:admit}
also fulfill the constraint on the maximal route length,
$\ALG$ simply uses preprocessing. We maintain at each node its relative load.
When a new request arrives $\ALG$ has to test
the costs of at most $O(n^\ell)$ chains, and the cost can be computed in $O(\ell)$
time per chain. The overall runtime of $\ALG$
per step is hence bounded by $O(\ell\cdot n^{\ell})$, which is polynomial
for constant $\ell$.


\subsection{Analysis}\label{ssec:stretch}

The analysis of the competitive ratio achieved by $\ALG$
exploits a connection to Virtual Circuit Routing~\cite{Plotkin95}
and unfolds in three lemmata. First, in Lemma~\ref{lemma:capacity}
we prove that
the set $A$ of requests admitted by $\ALG$ are feasible and respect capacity constraints.
Second, in Lemma~\ref{lemma:online}, we show that at any moment in time,
the sum of node costs is
within a factor $O({\ell}\cdot \log \mu)$
of the number of requests already admitted by $\ALG$.
Third, in Lemma~\ref{lemma:offline}, we prove that
the number of requests admitted by the optimal offline algorithm $\OFF$
but rejected by the online algorithm,
is bounded by the sum of node costs after processing all requests.

Let $W$  be the sum of the node costs after $\ALG$ processed
all $k$ request, let $A_{\OFF}$ be the indices of the requests admitted by $\OFF$,
and let $A^* = A_{\OFF} \setminus A$.
Then, from Lemma~\ref{lemma:online} we will obtain a bound
$|A|\geq W /( 2{\ell}\cdot \log \mu)$,
and from Lemma~\ref{lemma:offline} that $|A^*| \leq W/\ell$.

Thus, even by conservatively ignoring all the requests which $\ALG$
might have admitted which $\OFF$ did not, we obtain
that the competitive ratio of $\ALG$ is at most $O(\log \ell)$.

Let us now have a closer look at the first helper lemma.
\begin{lemma} \label{lemma:capacity}
For all nodes $v\in V$:
$$
 \sum_{j\in A: v\in c_j} 1 \leq \kappa(v).
$$
\end{lemma}
\begin{proof}
Let $\sigma_j$ be the first request admitted by $\ALG$,
such that the relative load $\lambda_v(j+1)$ at some node
$v\in c_j$ exceeds~1. By definition of the relative load we have
$\lambda_v(j) > 1- 1/\kappa(v)$.

By the assumption that $\log \mu \leq \kappa(v)$, we get
\begin{eqnarray*}
\frac{w_v(j)}{\kappa(v)} = \mu^{\lambda_v(j)} -1
 > \mu^{1-1/\log\mu}-1
 = \mu/2-1 = {\ell}.
\end{eqnarray*}
Therefore, by Condition~(\ref{eqn:admit}), the request $\sigma_j$ could not be assigned
to~$c_j$. We established a contradiction.
\qed
\end{proof}

Next we show that the sum of node costs is
within an $O({\ell}\cdot \log \mu)$ factor
of the number of already admitted requests.
\begin{lemma} \label{lemma:online}
Let $A$ be the set of indices of requests admitted by the online algorithm.
Let $k$ be the index of the last request. Then
$$
  (2 {\ell} \log\mu)|A| \geq \sum_v w_v(k+1).
$$
\end{lemma}
\begin{proof}
We show the claim by induction on $k$. For $k = 0$, both sides of the inequality
are zero, thus the claim is trivially true.
Rejected requests do not change either side of the inequality.
Thus, it is enough to show that, for each $j\leq k$, if we admit $\sigma_j$,
we get:
$$
\sum_v(w_v(j+1)-w_v(j)) \leq 2{\ell}\log\mu.
$$
Consider a node $v\in c_j$. Then by definition of the costs:
\begin{eqnarray*}
w_v(j+1)-w_v(j) & &= \kappa(v)(\mu^{\lambda_v(j)+1/\kappa(v)} - \mu^{\lambda_v(j)})\\
& &= \kappa(v)(\mu^{\lambda_v(j)}(\mu^{1/\kappa(v)} - 1))\\
& &= \kappa(v)(\mu^{\lambda_v(j)}(2^{(\log\mu) \cdot 1/\kappa(v)} - 1))
\end{eqnarray*}
By Assumption~(\ref{label:mincap}), $1\leq \kappa(v)/\log\mu$.
Since $2^x - 1 \leq x$, for $0\leq x\leq 1$, it follows:
\begin{eqnarray*}
w_v(j+1)-w_v(j) & &\leq \mu^{\lambda_v(j)} \log\mu\\
& & = \log\mu (w_v(j)/\kappa(v) +1).
\end{eqnarray*}
Summing up over all the nodes and using the fact that the request $\sigma_j$ was admitted
and chain $c_j$ was assigned, and that the number of nodes $|c_j|$ in $c_j$ is ${\ell}$,
we get:
$$
\sum_v(w_v(j+1)- w_v(j)) \leq \log \mu ({\ell} + |c_j|) = 2{\ell}\log\mu.
$$
This proves the claim.
\qed
\end{proof}

We finally prove that ${\ell}$ times the number of requests
rejected by $\ALG$ but admitted by the optimal offline algorithm
$\OFF$ is bounded by the sum of node costs after processing all requests.

\begin{lemma} \label{lemma:offline}
 Let $A_{\OFF}$ be the set of indices of the requests that were admitted
 by the optimal offline algorithm, and let $A^* = A_{\OFF} \setminus A$ be the
 set of indices of requests admitted by $A_{\OFF}$ but rejected by the
 online algorithm.
 Then:
 $$
 |A^*|\cdot {\ell} \leq \sum_{v} w_v(k+1).
$$
\end{lemma}

\begin{proof}
For $j\in A^*$, let $c^*_j$ be the chain assigned to request $\sigma_j$
by the optimal offline algorithm. By the fact that $\sigma_j$ was rejected by the
online algorithm, we have:
$$
{\ell} < \sum_{v\in c^*_j}\frac{w_v(j)}{\kappa(v)}.
$$
Since the costs $w_v(j)$ are monotonically increasing in $j$, we have
$$
{\ell} < \sum_{v\in c^*_j}\frac{w_v(j)}{\kappa(v)} \leq \sum_{v\in c^*_j}\frac{w_v(k+1)}{\kappa(v)}.
$$
Summing over all $j\in A^*$, we get
\begin{eqnarray*}
|A^*|{\ell} & &\leq \sum_{j\in A^*} \sum_{v\in c^*_j} \frac{w_v(k+1)}{\kappa(v)}\\
& & \leq \sum_{v} w_v(k+1)\cdot \sum_{j\in A^*: v\in c^*_j} \frac{1}{\kappa(v)}\\
& & \leq \sum_{v} w_v(k+1).
\end{eqnarray*}
The last inequality follows from the fact that
capacity constraints need to be met at any time.
\qed
\end{proof}

\begin{theorem} \label{theo:online}
$\ALG$ is $O(\log{\ell})$-competitive.
\end{theorem}
\begin{proof}
By Lemma~\ref{lemma:capacity}, capacity constraints are never violated.
It remains to show that the
number of requests admitted by the online algorithm is at least
$1/(2\log{2\mu})$ times the number of requests admitted by
the optimal offline algorithm.
The number of requests admitted by the optimal offline algorithm $|A_{\OFF}|$
can be bounded by the number of requests admitted by the online algorithm $|A|$
plus the number of requests in $A^* = A_{\OFF}\setminus A$. Therefore,
$$
|A_{\OFF}| \leq |A| + |A^*|.
$$
By Lemma~\ref{lemma:offline} this is bounded by
$$
|A_{\OFF}| \leq |A| + \frac{1}{\ell}\sum_v w_v(k+1).
$$
By Lemma~\ref{lemma:online} this is bounded by
\begin{eqnarray*}
|A_{\OFF}| & &\leq |A| + 2\cdot (\log\mu)  \cdot|A|\\
& & = (1+2\log\mu)|A|
\end{eqnarray*}
Therefore, the number of requests admitted by the optimal offline algorithm
is at most $(1+2\log\mu)$ times the number of requests admitted by
$\ALG$.
\qed
\end{proof}

\textbf{Remarks.}
We conclude with some remarks.
First, we note that our approach leaves us with many flexibilities
in terms of constraining the routes through the network functions.
For instance, we can support maximal path length requirements:
the maximal length of the route from $s$ to $t$ \emph{via the network functions}.
A natural alternative model is to define
a limit on the \emph{stretch}: the factor by which the ``detour'' via
the network functions can be longer than the shortest path from $s$ to $t$.
Moreover, so far, we focused on a model where requests, once admitted,
stay forever. Our approach can also be used to support service chain
requests of bounded or even unknown duration. In particular, by
redefining $\mu$ to take into account the duration of a request, we
can for example apply the technique from~\cite{Plotkin95} to obtain
competitive ratios for more general models.

\section{Optimality and Approximation}\label{sec:lowerbound}

It turns out that $\ALG$ is asymptotically optimal
within the class of 
online algorithms (Theorem~\ref{thm:lb-1}).
This section also initiates the study of lower bounds
for (offline) approximation algorithms, and shows
that for low capacities, the problem is APX-hard even for short chains (Theorem~\ref{thm:apx}),
and even Poly-APX-hard in general, that is, it is as hard as any problem that can be approximated to a polynomial factor
in polynomial time
(Theorem~\ref{thm:polyapx}).

\begin{theorem}\label{thm:lb-1}
Any deterministic or randomized online algorithm for OSCEP must have a competitive ratio
of at least $\Omega(\log{\ell})$.
\end{theorem}
\begin{proof}
We can adapt the proof of Lemma~4.1~in~\cite{AwerbuchAP93}
for our model.
We consider a capacity of $\kappa\geq\log{\ell}$, and we
divide the requests in $\sigma$ into $\log{\ell}+1$ phases.
We assume that $n\geq 2\ell^2$, and only focus
on a subset $L$ of $\ell=|L|$ nodes which are connected as a chain
$(v_1,\ldots,v_{\ell})$ and at which the different
service chains will overlap.
In phase 0, a group of $\kappa$ 
service chains are requested, all of which
need to be embedded across the nodes $L=\{v_1,\ldots,v_{\ell}\}$.
In phases $i\geq 1$, $2^i$ groups of $\kappa$ 
identical requests will need to share subsets of $L$ of size $\ell/2^i$,
that is, the $j$th group, $0\leq j< 2^i$, consists of $\kappa$ 
requests to be embedded across nodes $[v_{j\ell/2^i+1},v_{(j+1)\ell/2^i}]$.
See Figure~\ref{fig:lb} for an illustration.

\begin{figure}[ht]
\centering
\includegraphics[width=0.4\columnwidth]{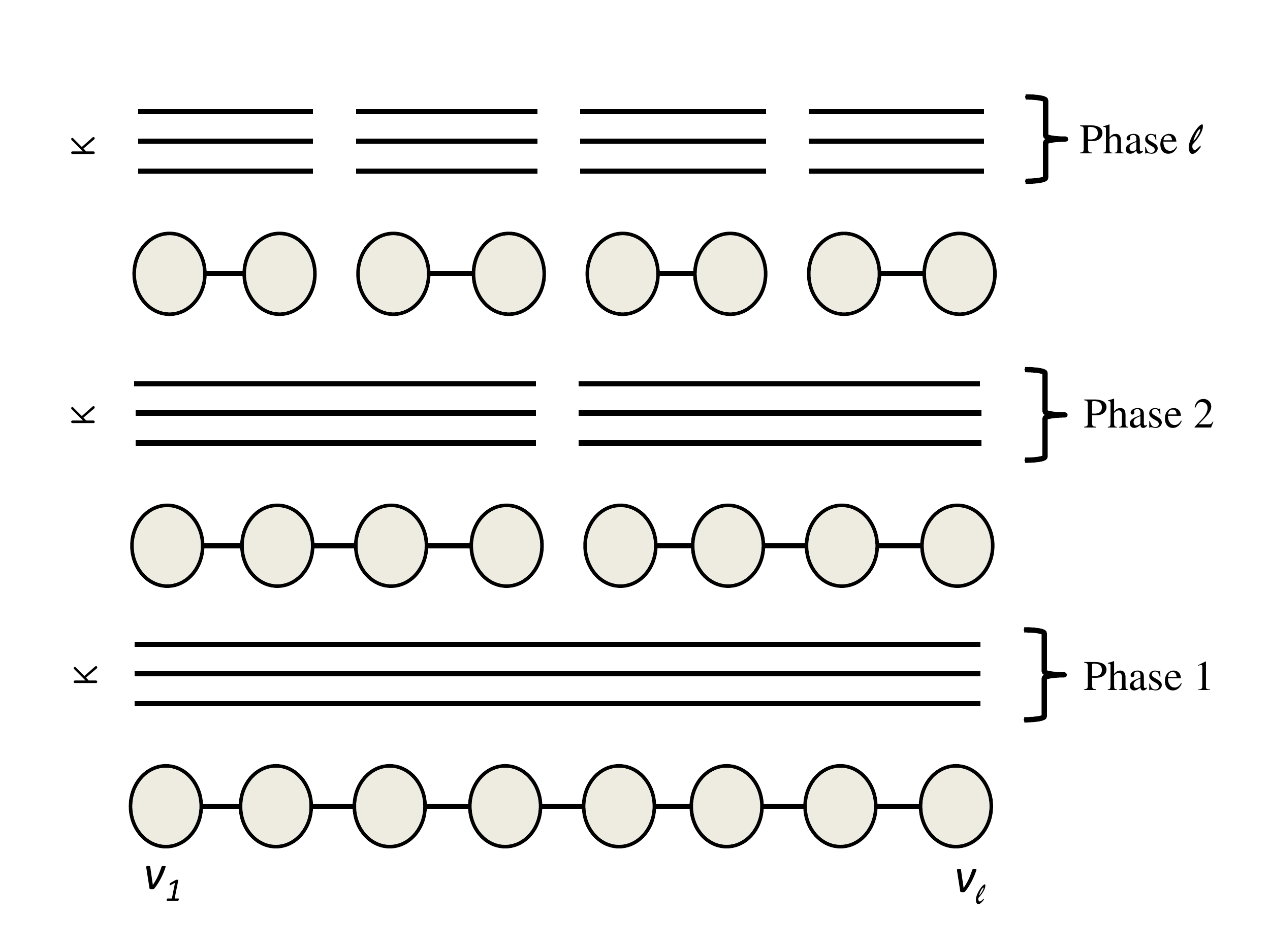}
\caption{Illustration of lower bound construction: The adversary issues
service chain requests in $1+\log \ell$ phases, where each phase $i$ consists of
$2^i$ 
groups of $\kappa\geq \log\ell$ requests. In phase 0 the adversary issues requests
that can be assigned to $L=(v_1,...,v_\ell)$.
As intersections of chains in phase $i$ with $L$ are becoming
shorter over time, the online algorithm needs to decide whether to admit service
service chain requests in phases, where each phase consists of groups with
$\kappa$ chains.
 As chains are becoming shorter over time,
the online algorithm faces the problem whether to admit service chains
early (and hence block precious resources), or late (in which case the adversary
stops issuing new requests).
}
\label{fig:lb}
\end{figure}

Let $x_i$ denote the number of requests an online algorithm \textbf{ON}
admits in phase $i$. Each request accepted in phase $i$ will occupy
$\ell/2^i$ units of capacities of nodes in $L$. Overall,
the nodes in $L$ have a capacity of
$\ell\cdot\kappa$,
so it must hold that
$$
\sum_{i=0}^{\log{\ell}} \frac{\ell}{2^i} \cdot x_i \leq \ell\cdot \kappa.
$$

Now, for $0\leq j\leq \log\ell$, define
$S_j= \frac{\ell}{2^j} \cdot\sum_{i=0}^{j} x_i.$
$S_j$ is a lower bound on the occupied capacity on the nodes of $L$ after phase~$j$.
Then:
\begin{eqnarray*}
\sum_{j=0}^{\log{\ell}} S_j & & =  \sum_{j=0}^{\log\ell} \frac{\ell}{2^j} \sum_{i=0}^j x_i\\
& & = \sum_{i=0}^{\log\ell} x_i \sum_{j=i}^{\log\ell} \frac{\ell}{2^j}\\
& & \leq  \sum_{i=0}^{\log\ell} x_i  2\frac{\ell}{2^i}\\
& & = 2\ell\kappa.
\end{eqnarray*}

Hence there must exist a $j$ such that $S_j \leq 2\ell\kappa/\log\ell$.
Then after phase $j$, the number of requests admitted by the online algorithm \textbf{ON}
is
$$
\sum_{i=0}^j x_i = \frac{2^j}{\ell}S_j\leq \frac{2^j}{\ell} 2 \ell\kappa / \log\ell = 2\cdot 2^j\kappa /\log\ell.
$$
The optimal offline algorithm \textbf{OFF} can reject all requests
except for those of phase $j$.  The number of requests in phase $j$, and thus, the number of
requests admitted by \textbf{OFF} is $2^j\kappa$.
\qed
\end{proof}

In the following, we also show that
for networks with low capacities,
it is not even possible to \emph{approximate} the offline version of the
Service Chain Embedding Problem, SCEP in polynomial time.
These lower bounds on the approximation ratio
naturally also constitute lower bounds on the competitive
ratio which can be achieved for OSCEP by any online algorithm.

In particular, we first show that
already for short chains in scenarios with unit capacities,
SCEP cannot be approximated well.
\begin{theorem}\label{thm:apx}
In scenarios where service chains have length $\ell\geq 3$ and
where capacities are
$\kappa(v)=1$, for all $v$, the offline problem is APX-hard.
\end{theorem}
\begin{proof}
The proof follows from an approximation-preserving reduction
from \emph{Maximum $k$-Set Packing Problem (KSP)}. 
The \emph{Maximum Set Packing (SP)} is one of Karp's 21 NP-complete problems,
where for a given collection $C$ of finite sets
a collection of disjoint sets $C'\subseteq C$ of maximum cardinality has to be found.
The KSP is the variation of the SP in which the cardinality of all sets in $C$ are bounded
from above by any constant $k\ge 3$, is APX-complete \cite{set-packing}.
We refer to such sets as $k$-sets.

KSP can be reduced to our problem as follows. Let $U$ be the universe and
$C$ be a collection of $k$-sets of $U$ in the KSP.
W.l.o.g., we assume that each $k$-set contains exactly $k$ elements, otherwise we can add
disjoint auxiliary elements to the sets in order to obtain exactly $k$ elements in each set
in~$C$.
For each $u\in U$ in the KSP instance we construct a node $v_u$ in the SCEP instance.
Furthermore, for each $k$-set $S$ in $C$, we construct a service chain $c_S$, such that
$c_S$ contains exactly the nodes $\{v_u\ :\ u\in S\}$. Let $\mathcal{C}$ be the set of
obtained service chains.
For the set of requests $\sigma$ we require that $|\sigma|\geq |\mathcal{C}|$ and that
each request can be assigned to each service chain.
Due to the unit capacity assumption, the set of admitted request must be assigned to
mutually disjoint service chains. Thus, the maximum number of admitted requests
is at most the maximum number of disjoint service chains.
Since each request can be assigned to each service chain and $|\sigma|\geq |\mathcal{C}|$,
an optimal solution for the SCEP determines a maximum set of mutually disjoint service chains.
This maximum set of disjoint service chains  determines a maximum number of disjoint $k$-sets,
and thus, an optimal solution for the KSP.
\qed\end{proof}

It turns out that in general, with unit capacities,
SCEP cannot even be approximated within polylogarithmic
factors.
\begin{theorem}\label{thm:polyapx}
In general scenarios where capacities are
$\kappa(v)=1$, for all nodes $v$, and the chain length $\ell\geq 3$,
the SCEP is.
APX-hard and not approximable within  $\ell^\varepsilon $ for some
$\varepsilon >0$.
Without a bound on the chain length the SCEP with $\kappa(v)=1$, for all nodes $v$,
is Poly-APX-hard.
\end{theorem}
\begin{proof}
We reduce the \emph{Maximum Independent Set (MIS)} problem with maximum degree $\ell$
to the SCEP with capacity $\kappa(v)=1$, for all $v\in V$ and chain length $\ell$.
For graphs with bounded degree $\ell\geq 3$, the MIS is APX-complete~\cite{Papadimitriou91}
and cannot be approximated within  $\ell^\varepsilon$ for some
$\varepsilon >0$~\cite{Alon95}.
By our reduction we obtain the APX-hardness and non-approximability within
 $\ell^\varepsilon $ for some $\varepsilon >0$  for the SCEP.
In general, for graphs without degree bound, the MIS is Poly-APX-complete~\cite{Bazgan05},
i.e. it is as hard as any problem that can be approximated to a polynomial factor.
By our reduction we obtain that the SCEP without chain length bound
is Poly-APX-hard.

For an instance $G=(V,E)$ of the MIS problem with maximum degree $\ell$, we construct
an instance of the SCEP with capacity $\kappa=1$ and chain length $\ell$ as follows.
For each  node $v\in G$, let $c_v$ be the chain whose nodes correspond to the edges
in $G$ incident to $v$. If $\deg_G(v)<\ell$ then we complete the chain with
$\ell-\deg_G(v)$ unique  auxiliary nodes, in order to have $\ell$ nodes in the chain.
The chain set is $C=\{c_v\ : \ v\in G\}$.
For the set of requests $\sigma$, we require that $|\sigma|\geq |C|$ and
each request $\sigma_i\in \sigma$ can be assigned to each $c\in C$.
Assigning a $\sigma_i$ to a chain $c\in C$ fills the capacity of all nodes in $c$
and the capacity of all chains $c'\in C$ that contain a common node with~$c$.
Therefore, no further request $\sigma_j$, $j\neq i$, can be assigned to those chains.
The chains having a common node with $c_v$ correspond exactly the neighbors of $v$ in $G$.
Therefore, nodes $u$ and $v$ are independent in the MIS instance iff chains $c_u$ and $c_v$
do not have a common node in the SCEP instance.
Since each request $\sigma_i$ can be assigned to each $c\in C$ and $|\sigma|\geq |C|$,
a maximum number of admitted requests is determined by a maximum chain set $C'$,
such that for all $c_u,c_v\in C'$,
$c_u$ and $c_v$ do not contain a common node.
Therefore, $C'$ determines a maximum independent set in~$G$.
Consequently, an $\alpha$-approximation for the SCEP would
imply an $\alpha$-approximation for the MIS problem.
\qed\end{proof}



\section{Summary and Conclusion}\label{sec:summary}

Over the last decades, a large number of middleboxes have
been deployed in computer networks, to increase
security and application performance,
as well as to offer new services in the form of
static and dynamic in-network processing (see the services by Akamai, Google Global Cache, Netflix Open Connect).

However,
the increasing cost and inflexibility of hardware middleboxes
(slow deployment, complex upgrades, lack of scalability),
motivated the advent of Network Function Virtualization (NFV)~\cite{routebricks,opennf,modeling-middleboxes,clickos},
which aims to run the
functionality provided by middleboxes as software on commodity
hardware.
The transition to NFV is discussed
within standardization groups such as ETSI, and we currently also
witness first
deployments, e.g., TeraStream~\cite{terastream}.

The possibility to chain individual network functions to form more complex
services has recently attracted much interest, both in academia~\cite{karl-chains,merlin},
as well as in industry~\cite{ewsdn14}.

Our paper made a first step towards a better understanding
of the algorithmic problem underlying the embedding
of service chains.
Our main contribution is a deterministic online algorithm $\ALG$
which achieves a competitive ratio of
$O(\log{\ell})$ for OSCEP, given that node capacities
are at least $\Omega(\log{\ell})$. This is interesting,
as the number $\ell$ of to-be-chained network functions is
likely to be a small constant in practice.
We also show that $\ALG$ is asymptotically optimal, in the sense that no
deterministic or randomized \emph{online} algorithm can achieve a competitive ratio
$o(\log{\ell})$. Moreover, we initiate the study of lower bounds
for the offline version of our problem, and show that no good
approximation algorithms exist, unless
$P=NP$:
the offline problem is APX-hard for unit capacities and service chains of length three.
In general,
the problem is even Poly-APX-hard under unit capacities. These results
also apply to the offline version of classic Virtual Circuit Routing.
Finally, this paper presented an exact algorithm based on 0-1 linear programming
for solving the offline SCEP optimally,
which implies that the offline SCEP is in NP,
if the size of the 0-1 program is polynomial, which holds for constant $\ell$
 --  0-1 programming is one
of  Karp's 21 NP-complete problems~\cite{Karp72}.

We believe our paper opens several interesting directions
for future research. For instance, it would be interesting
to find a lower bound for the approximation ratio
for the offline problem version where $\ell=2$.


{\bibliographystyle{abbrv}
\small
\bibliography{typeinst}
}

\begin{appendix}

\end{appendix}

\end{document}